\documentclass[letterpaper, 10pt, conference]{ieeeconf}

\IEEEoverridecommandlockouts
\usepackage[utf8]{inputenc}
\usepackage[T1]{fontenc}

\usepackage{graphicx}

\usepackage{algorithm}
\usepackage{algorithmic}

\usepackage{amsmath}
\usepackage{amssymb}
\usepackage{mathtools}
\usepackage{sgame}

\usepackage{booktabs}

\usepackage{float}
\usepackage{xcolor}

\usepackage{cite}
\usepackage[capitalize]{cleveref}

\usepackage{enumerate}

\newtheorem{remark}{Remark}
\newtheorem{proposition}{Proposition}
\newtheorem{theorem}{Theorem}

\newtheorem{definition}{Definition}
\newtheorem{lemma}{Lemma}
\newtheorem{assumption}{Assumption}

\DeclareMathOperator{\Equaldef}{\overset{\mathrm{def}}{=}}

\title{\LARGE \bf
Graphon Design for Human-Machine Coordination 
under Bounded Rationality: Optimality of Stochastic Block Models}

\author{Zhewei Wang, Vu Anh Phi, and Marcos M. Vasconcelos%
\thanks{Z. Wang is with the Department of Mechanical Engineering, V. A. Phi and M.\,M. Vasconcelos are with the Department of Electrical
and Computer Engineering, FAMU-FSU College of Engineering,
Florida State University, Tallahassee, FL 32306, USA.
E-mails: \texttt{zw23a@fsu.edu,vap25@fsu.edu, m.vasconcelos@fsu.edu}.}}

\begin{document}

\maketitle
\thispagestyle{empty}
\pagestyle{empty}

\begin{abstract}

Coordination is a desirable feature in multi-agent systems, ranging from robotic swarms to socioeconomic networks. This paper is concerned with promoting coordination  among heterogeneous agents, e.g., machines and humans, interacting in a stag-hunt game. In our model the agents exhibit bounded rationality at different levels, which leads to uncertainty and a propensity for errors during learning and decision-making processes. This paper addresses the problem of designing a network topology that maximizes a global metric of coordination under such constraints. While optimizing over the discrete space of finite graphs is generally computationally intractable, we employ a mean-field approach to lift the problem into the space of graphons. Within this framework, we analyze agents following a logit learning dynamics. Using calculus of variations, we show that for systems with a bimodal rationality profile, it suffices to search for optimal graphons in the ensemble of stochastic block models. We then propose a \textit{water-filling} algorithm to find a locally optimal graphon. Finite graphs can then be sampled from the optimized graphon, bypassing the inherent combinatorial complexities of discrete graph optimization.

\end{abstract}

\section{Introduction}

Coordination is a recurring theme in distributed decision-making, with broad applications across robotics, sensing, and communication networks. Traditionally, coordination in engineering has been studied through the lens of consensus protocols, where agents continuously seek to align their decision variables. In contrast, researchers in economics study coordination through the resolution of \emph{social dilemmas}. The most famous example is the Prisoner's Dilemma (PD) \cite{BAHEL2022126}, a matrix game where decision-makers must choose whether to cooperate or defect. While the sole Nash equilibrium of the PD is mutual defection, a global reward incentive exists for cooperationâ€”hence the \textit{dilemma}. Although networked versions of the PD are widely studied across engineering and social sciences, the PD is not considered a coordination game because it has a single equilibrium.

An arguably richer model for networked coordination is the \textit{Stag Hunt} game \cite{skyrms2004stag}. In this game, which is characterized by two pure-strategy Nash equilibria, agents must decide whether to collaborate on a risky but highly rewarding task (hunting stag) or settle for a safe, low-rewarding alternative (hunting hare). The networked version of this game connects engineering, economics, and statistical physics, providing a rich framework to study the influence of network topology on the emergence of globally coordinated behavior. When considering a mix of human and machine agents, it is natural to assume agents with different \textit{rationality} levels. Here, we consider a rational agent as one that maximizes its expected utility, an assumption that is often taken for granted in the analysis of economic, artificial intelligence, and engineering systems \cite{recht2026irrational}.

Achieving effective coordination in hybrid human-machine networks is one of the fundamental challenges in modern multi-agent systems \cite{tsvetkova2024new}. For instance, Crandall et al. \cite{crandall2018cooperating} studied the inherent difficulties in sustaining long-term cooperation between humans and machines. Shirado and Christakis \cite{shirado2017locally} demonstrated that introducing autonomous agents can help large human groups overcome local deadlocks and promote global coordination. While these empirical results highlight the value of embedding autonomous agents within human networks, the literature currently lacks a rigorous mathematical framework to formalize this phenomenon and help us allocate the appropriate links between humans and machines to promote coordinated behavior thereby compensating for bounded and heterogeneous rationality levels. 

One possible approach is to use the log-linear learning formalism introduced in \cite{blume1993statistical} and adapt it to handle agents with heterogeneous rationality \cite{Rajab:2025}. 
However, the theoretical treatment of network optimization for heterogeneous systems quickly becomes intractable due to two challenges: (1) the combinatorial explosion of the discrete graph space, which is computationally prohibitive even for a modest number of agents (e.g., $N \approx 20$) \cite{zhang2024rationality,zhang2024role}; and (2) the absence of a closed-form stationary probability distribution when rationality is heterogeneously distributed among agents, which effectively precludes  optimization over graphs \cite{Rajab:2025}.

In an effort to overcome these fundamental difficulties, we adopt a mean-field approach based on \textit{graphons} \cite{parise2023graphon} -- continuous objects that represent the limit of large-scale networks for a continuum of agents. By relaxing the problem from a finite agent set to a continuum population, we circumvent the inherent combinatorial complexity of the network design problem. Rather than attempting to optimize the stationary distribution of a discrete Markov chain induced by logit learning \cite{marden2012revisiting}, we reformulate the objective to maximize the expected fraction of agents that asymptotically play the equilibrium that maximizes the game's potential function subject to a constraint on the admissible graphons. Our approach entirely bypasses the need for a closed-form stationary distribution used in \cite{zhang2024rationality,zhang2024role,wang2026optimizing}.

Here, we focus on network design for hybrid systems characterized by a bimodal rationality profile. 
To address this, we provide an analytical characterization of the optimal topology for a system with two types of agents. By solving the optimization problem in the continuous graphon space, we show that it suffices to restrict our search for optimal graphons to the stochastic block model (SBM) ensemble \cite{holland1983stochastic, abbe2018community}. We then propose a water-filling algorithm to compute the locally optimal graphon profile.

The rest of the paper is organized as follows. \Cref{sec:setup} introduces the problem setup and the graphon coordination game. \Cref{sec:logit} formulates the logit learning dynamics for heterogeneous rationality, defines the aggregate adoption objective, and establishes the uniqueness of the steady state under a contraction condition. \Cref{sec:optimality} derives the first-order optimality conditions via a Lagrangian approach. \Cref{sec:bimodal} specializes these conditions to a bimodal rationality profile and proves that every local maximizer is a stochastic block model. \Cref{sec:allocation} develops a greedy connectivity allocation algorithm. \Cref{sec:numerical} illustrates the algorithm numerically. \Cref{sec:conclusions} discusses extensions and open problems.

\section{Problem Setup}
\label{sec:setup}

Consider a population of unit mass $\mathcal{P} = [0,1]$
partitioned into two types. A fraction $\ell \in (0,1)$ are
\emph{humans}, indexed by $x \in [0,\ell)$, and a fraction $1-\ell$
are \emph{machines}, indexed by $x \in [\ell,1]$. Agents of the
same type are identical. The interaction structure among the agents is described by a
\emph{graphon} $W:[0,1]^2\to[0,1]$, a symmetric and measurable
function encoding the connection probability between any two
agents. We restrict to graphons with a fixed edge density
$\rho \in (0,1)$. Let the set of graphons of density $\rho$ be defined as
\begin{multline}
  \mathcal{W}_\rho \;\Equaldef\;
  \Big\{
    W \in L^\infty([0,1]^2)
   \mid 
    W(x,y) = W(y,x) \in [0,1],\;  \\
    \int_0^1\int_0^1 W(x,y)\,dx\,dy = \rho
  \Big\}.
\end{multline}

\subsection{Graphon Coordination Game}

Fix a difficulty parameter $\theta \in (0,1)$. Each pair of agents
$x, y \in \mathcal{P}$ plays a binary \textit{stag hunt} game: $1$
represents the risky action (hunting stag) and $0$ represents
the safe action (hunting hare) \cite{skyrms2004stag,Montanari:2010}. 

\begin{figure}[h!]
  \hspace*{\fill}%
  \centering
  \begin{game}{2}{2}[$a(x)$][$a(y)$]
       & $1$ & $0$  \\
   $1$ & $\bigl(1-\theta,\,1-\theta\bigr)$ & $\bigl(-\theta,\,0\bigr)$ \\
   $0$ & $\bigl(0,\,-\theta\bigr)$         & $\bigl(0,\,0\bigr)$       \\
  \end{game}%
  \hspace*{\fill}%
  \caption{Stag hunt coordination game between agents $x, y \in \mathcal{P}$
         with difficulty parameter $\theta \in (0,1)$.}
  \label{fig:bimatrix}
\end{figure}

Let $\mathcal{A}$ denote the set of pure strategy profiles,
\begin{equation}
    \mathcal{A} = \bigl\{\, a \in L^1([0,1]) : [0,1] \to \{0,1\}\,\bigr\}.
\end{equation}
The game's payoff in \cref{fig:bimatrix} can be compactly written as
\begin{equation}
  u_x\bigl(a(x),\,a(y)\bigr)
  = a(x)\bigl(a(y) - \theta\bigr).
\end{equation}
Because each agent plays the same action with the entire population,
the utility is aggregated and weighted by the graphon, which gives
\begin{equation}
  \label{eq:utility}
  u_x\bigl(a(x),\,a\bigr)
  = a(x)\int_0^1 W(x,y)\,\bigl(a(y) - \theta\bigr)\,dy.
\end{equation}

\begin{definition}[Nash Equilibrium]
A strategy profile $a^\star \in \mathcal{A}$ is a
\emph{Nash equilibrium} if, for all alternative actions $s \in \{0,1\}$,
\begin{equation}
  u_x\bigl(a^\star(x),\,a^\star\bigr)
  \;\geq\;
  u_x\bigl(s,\,a^\star\bigr)
  \quad \text{a.e.}
\end{equation}
\end{definition}

\vspace{5pt}

\begin{proposition}
\label{prop:consensus_equilibria}
Let $W : [0,1]^2 \to [0,1]$ be a graphon and let $\theta \in (0,1)$.
For the graphon game defined by \eqref{eq:utility}, the constant strategy
profiles $a^\star(x) = 0$ a.e.\ and $a^\star(x) = 1$ a.e.\ are both
Nash equilibria.
\end{proposition}

\vspace{5pt}

\begin{proof}
The proof is immediate and omitted for brevity.
\end{proof}

\subsection{Potential Game Structure}

A key structural property of our graphon coordination game is that it
is a \emph{potential game} \cite{monderer1996potential}.
The potential function is
\begin{multline}
  \label{eq:potential}
  \Phi_{W}(a)
  \;\Equaldef\;
  \frac{1}{2}\int_0^1\!\int_0^1 W(x,y)\,a(x)\,a(y)\,dx\,dy
  \\
  -\;\theta\int_0^1 a(x)\,D(x)\,dx,
\end{multline}
where the \emph{connectivity degree} of agent $x$ is
\begin{equation}
  D(x) \;\Equaldef\; \int_0^1 W(x,y)\,dy.
\end{equation}

To verify that \eqref{eq:potential} is a potential function, we can show that the G\^{a}teaux derivative \cite{liberzon2012calculus} of $\Phi_W$ with respect to $a$ at $x$ is  $u_x(1,a) - u_x(0,a)$. 

Although the graphon coordination game may admit many other equilibria than the
perfectly aligned profiles $a^\star(x) = 1$ and $a^\star(x) = 0$, the following proposition establishes that these two profiles are the \emph{only} global maximizers of
$\Phi_W$.

\vspace{5pt}

\begin{proposition}[Potential Maximizers]
\label{prop:potential_max}
Let $W$ be a symmetric graphon satisfying
$\int_0^1\!\int_0^1 W(x,y)\,dx\,dy > 0$, and let $\theta \in (0,1)$.
The global maximum of $\Phi_W(a)$ defined in \eqref{eq:potential} is
attained as follows
\begin{enumerate}[(i)]
  \item If $\theta < 1/2$, the unique maximizer is $a^\star(x) = 1$ a.e.
  \item If $\theta > 1/2$, the unique maximizer is $a^\star(x) = 0$ a.e.
  \item If $\theta = 1/2$, both $a^\star(x) = 1$ and $a^\star(x) = 0$ maximize
        $\Phi_W(a)$.
\end{enumerate}
\end{proposition}

\vspace{5pt}

\begin{proof}
The proof is immediate and omitted for brevity.
\end{proof}

\section{Logit Learning Dynamics}\label{sec:logit}

We are interested in optimizing the graphon $W$ to influence the behavior of boundedly rational agents in a mixed population of humans and machines. We assume that agents of each type have different levels of rationality. To model the evolution of behavior in the continuum population, we
introduce a continuous-time \emph{logit learning dynamic} \cite{sandholm2010population,Cianfanelli:2025}.

The space of mixed strategy profiles is denoted by
\begin{equation}
\mathcal{M} = \big\{a \in L^1([0,1])\ : [0,1] \to [0,1] \big\}.
\end{equation}
Let $a_t(x) \in [0,1]$ 
denote the probability that agent $x$ plays action $1$ at time $t \geq 0$. 
Under independent Poisson revision clocks, each agent updates its action according to the logit rule. Given the current population profile $a_t$, the incentive for agent
$x$ to play strategy $1$ over strategy $0$ is
\begin{align}
  \Delta u_x(a_t)
  &\;\Equaldef\; u_x(1, a_t) - u_x(0, a_t) \notag \\
  &= \int_0^1 W(x,y)\,\bigl(a_t(y) - \theta\bigr)\,dy.
  \label{eq:difference}
\end{align}

Under the logit learning dynamic, agents do not best-respond 
perfectly; instead, they choose actions with probabilities 
proportional to their expected utility, governed by a type dependent
rationality parameter $\beta(x) \in [0,\infty)$ according to the logit kernel
\begin{align}\label{eq:logit_dynamic}
  \sigma_{\beta(x)}\!\bigl(\Delta u_x(a_t)\bigr)
  &\Equaldef \frac{1}{1 + \exp\bigl(-\beta(x)\,\Delta u_x(a_t)\bigr)}.
\end{align}
As $\beta(x) \to 0$, agents choose uniformly at random; as
$\beta(x) \to \infty$, they best-respond perfectly.

In our model, the rationality profile $\beta(x)$ is assumed to be of the following form
\begin{equation}
  \label{eq:beta_bimodal}
  \beta(x) =
  \begin{cases}
    \beta_H & x \in [0,\ell) \quad \text{(humans)} \\
    \beta_M & x \in [\ell,1] \quad \text{(machines)}
  \end{cases}
\end{equation}
where $0 < \beta_H < \beta_M$. Following a standard argument from the population games literature \cite{sandholm2010population}, the mixed action profile
evolves as
\begin{equation}
  \frac{\partial a_t(x)}{\partial t}
  = \sigma_{\beta(x)}\!\left(
      \int_0^1 W(x,y)\,\bigl(a_t(y)-\theta\bigr)\,dy
    \right)
  - a_t(x).
\end{equation}
Since the underlying game is a potential game, the logit dynamics
\eqref{eq:logit_dynamic} admits a strict Lyapunov functional. As shown in \cite{Cianfanelli:2025}, this guarantees asymptotic
convergence to a steady-state profile $a_\infty(x)$.
Setting $\partial a_t(x)/\partial t = 0$ yields
\begin{equation}
  \label{eq:fixed_point}
  a_\infty(x)
  = \sigma_{\beta(x)}\!\left(
      \int_0^1 W(x,y)\,\bigl(a_\infty(y)-\theta\bigr)\,dy
    \right)
  \quad \text{a.e.}
\end{equation}

\begin{remark}
\Cref{eq:fixed_point} is akin to a mean-field version of a \emph{Quantal Response Equilibrium} \cite{GoereeHoltPalfrey2016}, which is a generalized version of a Nash Equilibrium for agents exhibiting bounded rationality.
\end{remark}

\subsection{Optimization Problem}

For $\theta < 1/2$, the action profile $a^\star(x) = 1$ a.e.\ is the 
socially desirable equilibrium by \cref{prop:potential_max}. Therefore, 
a natural performance metric is the asymptotic adoption fraction
\begin{equation}
  J_1(W) \;\Equaldef\; \int_0^1 a_\infty(x)\,dx,
\end{equation}
whose maximization is equivalent to minimizing the $L^1$ distance between 
$a^\star$ and $a_\infty$. To see this, note that since $a^\star(x) = 1$ a.e.\ and 
$a_\infty(x) \in (0,1)$, we have
\begin{equation}
  \|a^\star - a_\infty\|_{1} 
  = \int_0^1 |1 - a_\infty(x)|\,dx 
  = 1 - J_1(W).
\end{equation}

For the remainder of the paper, we will consider the following graphon optimization problem 
\begin{equation}
  \label{eq:opt}
  \max_{W \in \mathcal{W}_\rho}\;
  J_1(W)
\end{equation}
subject to the nonlinear equality constraint imposed by \eqref{eq:fixed_point}.

\subsection{Uniqueness of the steady state}

Although the potential game structure guarantees convergence to a fixed point, the logit dynamics is notorious for their multiplicity when $\beta(x)$ is too large \cite{Cianfanelli:2025}.
We impose the following condition to guarantee a unique steady
state $a_\infty$ for every $W \in \mathcal{W}_\rho$ and the well-posedness of the optimization problem.

\vspace{5pt}

\begin{assumption}
\label{ass:contraction}
The rationality profile $\beta:[0,1]\to(0,\infty)$ and
graphon $W \in \mathcal{W}_\rho$ satisfy
\begin{equation}
  \label{eq:contraction_condition}
  \sup_{x \in [0,1]}\,\beta(x)\,D(x) < 4,
\end{equation}
where $D(x) = \int_0^1 W(x,y)\,dy$ is the degree of
agent $x$.
\end{assumption}

\vspace{5pt}

\begin{proposition}
\label{prop:uniqueness}
Under \cref{ass:contraction}, the operator
\begin{equation}
  \mathcal{F}(a)(x)
  \;\Equaldef\;
  \sigma_{\beta(x)}\!\left(
    \int_0^1 W(x,y)\bigl(a(y)-\theta\bigr)\,dy
  \right)
\end{equation}
is a contraction on $(\mathcal{M},\|\cdot\|_{\infty})$
with constant
\begin{equation}
  \kappa
  \;\Equaldef\;
  \frac{1}{4}\sup_{x \in [0,1]}\beta(x)\,D(x)
  \;<\; 1.
\end{equation}
In particular, \eqref{eq:fixed_point} has a unique solution
$a_\infty \in \mathcal{M}$.
\end{proposition}

\vspace{5pt}

\begin{proof}
Since $\sigma_{\beta(x)}$ maps $\mathbb{R}$ into
$(0,1) \subset [0,1]$ for any $\beta(x) > 0$, we have
$\mathcal{F}(a)(x) \in (0,1)$ for all $x$, so
$\mathcal{F}:\mathcal{M}\to\mathcal{M}$. Next, we derive the Lipschitz constant of $\sigma_{\beta(x)}$.
Since
\begin{equation}
  \sigma_{\beta(x)}'(z)
  = \beta(x)\,\sigma_{\beta(x)}(z)
    \bigl(1-\sigma_{\beta(x)}(z)\bigr),
\end{equation}
and $s(1-s) \leq 1/4$ for all $s \in [0,1]$, we obtain
\begin{equation}
  \bigl|\sigma_{\beta(x)}'(z)\bigr|
  \leq \frac{\beta(x)}{4},
  \qquad z \in \mathbb{R},\; x \in [0,1].
\end{equation}
Therefore, for any $a, b \in \mathcal{M}$ and
a.e.\ $x \in [0,1]$, we have
\begin{align}
  \bigl|\mathcal{F}(a)(x) - \mathcal{F}(b)(x)\bigr|
  &\overset{(a)}{\leq}
    \frac{\beta(x)}{4}
    \left|\int_0^1 W(x,y)\bigl(a(y)-b(y)\bigr)\,dy\right|
    \notag \\
  &\overset{(b)}{\leq}
    \frac{\beta(x)}{4}
    \int_0^1 W(x,y)\,|a(y)-b(y)|\,dy \notag \\
  &\overset{(c)}{\leq}
    \frac{\beta(x)}{4}\,D(x)\,\|a-b\|_{\infty},
\end{align}
where $(a)$ follows from the Lipschitz bound on
$\sigma_{\beta(x)}$;
$(b)$ follows from the triangle inequality for integrals;
and $(c)$ follows from the definition of the $L^\infty$ norm together with the definition of degree.

Taking the supremum over $x \in [0,1]$, we have
\begin{equation}
  \|\mathcal{F}(a) - \mathcal{F}(b)\|_{\infty}
  \leq
  \underbrace{\frac{\sup_x\beta(x)\,D(x)}{4}}_{\kappa}
  \|a - b\|_{\infty}.
\end{equation}
Under \cref{ass:contraction}, $\kappa < 1$ and $\mathcal{F}$
is a contraction on $\mathcal{M}$. Since $\mathcal{M}$ is a
closed subset of $L^\infty([0,1])$, 
$(\mathcal{M}, \|\cdot\|_{\infty})$ is complete. The Banach
fixed-point theorem therefore guarantees a unique
$a_\infty \in \mathcal{M}$ satisfying
$\mathcal{F}(a_\infty) = a_\infty$.
\end{proof}

\vspace{5pt}

\begin{proposition}
\label{prop:uniqueness_simple}
If $\beta(x) < 4$ for all $x \in [0,1]$, then the
hypothesis of \cref{prop:uniqueness} holds for every
$W \in \mathcal{W}_\rho$, and consequently
\eqref{eq:fixed_point} has a unique solution
$a_\infty \in \mathcal{M}$ for every $W \in \mathcal{W}_\rho$.
\end{proposition}

\vspace{5pt}

\begin{proof}
Since $D(x) = \int_0^1 W(x,y)\,dy \leq \int_0^1 1\,dy = 1$
for any $W \in \mathcal{W}_\rho$, we have
\begin{equation}
  \sup_{x \in [0,1]} \beta(x)\,D(x)
  \leq  \sup_{x \in [0,1]} \beta(x)
  < 4.
\end{equation}
The result follows directly from \cref{prop:uniqueness}.
\end{proof}

\section{First-Order Optimality Conditions}
\label{sec:optimality}

Treating the optimization problem from the point of
view of calculus of variations \cite{liberzon2012calculus}, we proceed by
characterizing its first-order optimality conditions.

\subsection{Lagrangian and Optimality Conditions}

\begin{lemma}
\label{lem:first_order_condition}
Let $W^\star \in \mathcal{W}_\rho$ be a local maximizer
of $J_1(W) = \int_0^1 a_\infty(x)\,dx$, with associated
steady-state profile $a^\star$ satisfying
\eqref{eq:fixed_point} and local externality
$z^\star(x) = \int_0^1 W^\star(x,y)(a^\star(y)-\theta)
\,dy$. Let $\mu^\star:[0,1]\to\mathbb{R}$ be the unique
solution of the following integral equation:
\begin{equation}
  \label{eq:adjoint}
  \mu^\star(x)
  = \sigma_{\beta(x)}'\big(z^\star(x)\big)
    \!\left(
      1 + \int_0^1 W^\star(x,y)\,\mu^\star(y)\,dy
    \right)
  \quad \text{a.e.,}
\end{equation}
and define the symmetric function:
\begin{equation}
  \label{eq:switching}
  G^\star(x,y)
  \;\Equaldef\;
  \mu^\star(x)\big(a^\star(y)-\theta\big)
  + \mu^\star(y)\big(a^\star(x)-\theta\big).
\end{equation}
Then the following holds:
\begin{enumerate}[(i)]
  \item $\mu^\star(x) > 0$ for all $x \in [0,1]$.
  \item For every $V \in \mathcal{W}_\rho$, we have
  \begin{equation}
    \label{eq:vi}
    \int_0^1\!\int_0^1 G^\star(x,y)
    \bigl(V(x,y) - W^\star(x,y)\bigr)\,dx\,dy \leq 0.
  \end{equation}
  \item There exists $\lambda \in \mathbb{R}$ such that
  \begin{equation}
    \label{eq:bangbang}
    W^\star(x,y) =
    \begin{cases}
      1        & \text{if } G^\star(x,y) > 2\lambda \\
      \in[0,1] & \text{if } G^\star(x,y) = 2\lambda \\
      0        & \text{if } G^\star(x,y) < 2\lambda.
    \end{cases}
    \quad \text{a.e.}
  \end{equation}
\end{enumerate}
\end{lemma}

\begin{proof}
We enforce the fixed-point equation~\eqref{eq:fixed_point}
pointwise via a multiplier function $\mu(x)$ and the
density constraint via a scalar multiplier $\lambda$,
giving the Lagrangian functional
\begin{multline}
  \label{eq:lagrangian}
  \mathcal{L}(W, a, \mu, \lambda)
  \Equaldef \int_0^1 a(x)\,dx \\
  + \int_0^1 \mu(x)\!\left[
      \sigma_{\beta(x)}\!\left(
        \int_0^1 W(x,y)(a(y)-\theta)\,dy
      \right)
      - a(x)
    \right]dx \\
  + \lambda\!\left[\rho - \int_0^1\!\int_0^1
    W(x,y)\,dx\,dy\right].
\end{multline}
The stationarity of $\mathcal{L}$ with respect to $\mu(x)$
recovers the fixed-point equation~\eqref{eq:fixed_point},
and the stationarity with respect to $\lambda$ recovers
the density constraint $\iint W^\star\,dx\,dy = \rho$.

Taking the Fr\'{e}chet derivative of $\mathcal{L}$ with
respect to $a$ in the direction $\delta a \in L^\infty$, linearizing the logit function, and 
setting the first variation to zero, the symmetry of $W^\star$ yields
\begin{equation}
  \label{eq:mu_eq}
  \mu(x)
  = 1 + \int_0^1 \mu(y)\,
    \sigma_{\beta(y)}'\big(z^\star(y)\big)\,
    W^\star(x,y)\,dy
  \quad \text{a.e.}
\end{equation}
Under \cref{ass:contraction}, \eqref{eq:mu_eq} has a
unique solution, which we denote by $\tilde{\mu}(x)$. Since the integral term is
non-negative, $\tilde{\mu}(x) \geq 1 > 0$ for all $x$.
Defining the adjoint variable
\begin{equation}
  \mu^\star(x)
  \;\Equaldef\;
  \tilde{\mu}(x)\,\sigma_{\beta(x)}'\big(z^\star(x)\big)
\end{equation}
and multiplying both sides of~\eqref{eq:mu_eq} by
$\sigma_{\beta(x)}'\big(z^\star(x)\big)$ yields~\eqref{eq:adjoint}.
Positivity $\mu^\star(x) > 0$ follows from
$\tilde{\mu}(x) \geq 1$ and $\sigma_{\beta(x)}'(z) > 0$
for all $z \in \mathbb{R}$, proving~(i).

Perturbing $W^\star$ towards any $V \in \mathcal{W}_\rho$
gives
\begin{multline}
  \delta_W\mathcal{L}
  = \int_0^1\!\int_0^1
    \mu^\star(x)\big(a^\star(y)-\theta\big)
    (V-W^\star)(x,y)\,dy\,dx \\
  - \lambda\int_0^1\!\int_0^1
    (V-W^\star)(x,y)\,dx\,dy.
\end{multline}
Since $V - W^\star$ is symmetric, we can write
\begin{multline}
  \int_0^1\!\int_0^1
    \mu^\star(x)\big(a^\star(y)-\theta\big)
    (V-W^\star)(x,y)\,dy\,dx \\
  = \frac{1}{2}\int_0^1\!\int_0^1
    G^\star(x,y)\,(V-W^\star)(x,y)\,dx\,dy.
\end{multline}
Since $V, W^\star \in \mathcal{W}_\rho$, we have
$\iint(V - W^\star)\,dx\,dy = \rho - \rho = 0$, so
the term multiplying $\lambda$ vanishes, giving
\begin{equation}
  \delta_W\mathcal{L}
  = \frac{1}{2}\int_0^1\!\int_0^1
    G^\star(x,y)\,\big(V(x,y) - W^\star(x,y)\big)
    \,dx\,dy.
\end{equation}
Optimality of $W^\star$ requires $\delta_W\mathcal{L}
\leq 0$ for all $V \in \mathcal{W}_\rho$,
yielding~\eqref{eq:vi} and proving~(ii).

At fixed $(a^\star, \mu^\star, \lambda)$, the
Lagrangian is a linear functional of $W$:
\begin{multline}
  \tilde{\mathcal{L}}(W)
  = \frac{1}{2}\int_0^1\!\int_0^1
    G^\star(x,y)\,W(x,y)\,dx\,dy \\
  - \lambda\int_0^1\!\int_0^1 W(x,y)\,dx\,dy
  + \mathrm{constant}.
\end{multline}
Introducing pointwise multipliers $\alpha(x,y) \geq 0$
for $W(x,y) \geq 0$ and $\phi(x,y) \geq 0$ for
$1-W(x,y) \geq 0$, and differentiating with respect
to $W(x,y)$ pointwise yields the KKT stationarity
condition:
\begin{equation}
  \frac{G^\star(x,y)}{2} - \lambda
  + \alpha(x,y) - \phi(x,y)
  = 0 \quad \text{a.e.,}
\end{equation}
with complementary slackness
$\alpha(x,y)W^\star(x,y)=0$ and
$\phi(x,y)\big(1-W^\star(x,y)\big)=0$ a.e., meaning that each
multiplier is nonzero only when its associated
constraint is active. These conditions yield three
cases:
\begin{itemize}
  \item If $G^\star(x,y)/2 > \lambda$: stationarity
  requires $\phi(x,y) > \alpha(x,y) \geq 0$, so
  $\phi(x,y) > 0 \implies W^\star(x,y) = 1$.
  \item If $G^\star(x,y)/2 < \lambda$: stationarity
  requires $\alpha(x,y) > \phi(x,y) \geq 0$, so
  $\alpha(x,y) > 0 \implies W^\star(x,y) = 0$.
  \item If $G^\star(x,y)/2 = \lambda$: stationarity
  is satisfied by $\alpha(x,y) = \phi(x,y) = 0$
  for any $W^\star(x,y) \in [0,1]$.
\end{itemize}
This gives~\eqref{eq:bangbang}, proving~(iii).
\end{proof}

\section{Structure of the Optimal Graphon for
Bimodal Rationality}\label{sec:bimodal}

\Cref{lem:first_order_condition} established the
first-order optimality conditions for any rationality
profile satisfying $\beta(x) < 4$ for all
$x \in [0,1]$. We now specialize these conditions to
the bimodal rationality profile~\eqref{eq:beta_bimodal},
where $0 < \beta_H < \beta_M < 4$.

\vspace{5pt}

\begin{theorem}
\label{thm:sbm}
Let $\beta(x)$ be the bimodal rationality
profile~\eqref{eq:beta_bimodal} with
$0 < \beta_H < \beta_M < 4$, and let
$W^\star \in \mathcal{W}_\rho$ be a local maximizer
of $J_1(W)$. Then, the steady-state profile $a^\star(x)$
and adjoint $\mu^\star(x)$ are constant within each
agent type:
\begin{multline}
  \label{eq:constant_types}
  a^\star(x) =
  \begin{cases}
    a_H & \text{a.e. on } [0,\ell), \\
    a_M & \text{a.e. on } [\ell,1],
  \end{cases}
\\
  \mu^\star(x) =
  \begin{cases}
    \mu_H & \text{a.e. on } [0,\ell), \\
    \mu_M & \text{a.e. on } [\ell,1].
  \end{cases}
\end{multline}
Moreover, $W^\star$ is an SBM
graphon, i.e.,
\begin{equation}
  \label{eq:sbm}
  W^\star(x,y) =
  \begin{cases}
    w_{HH} & x,y \in [0,\ell), \\
    w_{HM} & (x,y) \in [0,\ell)\times[\ell,1]
             \cup [\ell,1]\times[0,\ell), \\
    w_{MM} & x,y \in [\ell,1],
  \end{cases}
\end{equation}
with block values $w_{HH}, w_{HM}, w_{MM} \in [0,1]$
determined by the
condition~\eqref{eq:bangbang}. The function
$G^\star(x,y)$ defined in~\eqref{eq:switching} takes
the following values on the corresponding blocks:
\begin{align}
  G_{MM} &\;\Equaldef\; 2\mu_M(a_M - \theta),
           \label{eq:gmm} \\
  G_{HM} &\;\Equaldef\; \mu_M(a_H-\theta)
           + \mu_H(a_M-\theta), \label{eq:ghm} \\
  G_{HH} &\;\Equaldef\; 2\mu_H(a_H - \theta).
           \label{eq:ghh}
\end{align}
\end{theorem}

\vspace{5pt}

\begin{proof}
Since $\beta(x)$ takes only two values, the population
is partitioned into two exchangeable groups: humans
$x \in [0,\ell)$ and machines $x \in [\ell,1]$. For
any measure-preserving permutation $\pi:[0,1]\to[0,1]$
satisfying $\pi([0,\ell))=[0,\ell)$ and
$\pi([\ell,1])=[\ell,1]$ a.e., the relabeled graphon
$W^\pi(x,y) \Equaldef W^\star(\pi(x),\pi(y))$
satisfies $J_1(W^\pi) = J_1(W^\star)$. Indeed, since
$\beta(\pi(x)) = \beta(x)$, the unique fixed point of
$W^\pi$ is $\tilde{a}(x) = a^\star(\pi(x))$ by
\cref{prop:uniqueness}, and measure-preservation gives
$J_1(W^\pi) = \int_0^1 a^\star(\pi(x))\,dx =
\int_0^1 a^\star(x)\,dx = J_1(W^\star)$.

Define the symmetrized graphon:
\begin{equation}
  \bar{W}(x,y)
  \;\Equaldef\;
  \int_\Pi W^\star\big(\pi(x),\pi(y)\big)\,d\nu(\pi),
\end{equation}
where $\Pi$ denotes the set of all such permutations
and $\nu$ is a probability measure on $\Pi$. By
construction, $\bar{W}(x,y)$ depends only on the types
of $x$ and $y$, so $\bar{W} \in \mathcal{W}_\rho$ is
an SBM graphon.

We now show $J_1(\bar{W}) \geq J_1(W^\star)$ using an
inductive argument on the fixed-point iteration. For
any fixed $a \in \mathcal{M}$ and a.e.\
$x \in [0,1]$, the map $W \mapsto \mathcal{F}_W(a)(x)
= \sigma_{\beta(x)}\!\left(\int_0^1
W(x,y)(a(y)-\theta)\,dy\right)$ is concave in $W$. Indeed,
the argument $z_W(x) = \int_0^1
W(x,y)(a(y)-\theta)\,dy$ is linear in $W$, and
$\sigma_{\beta(x)}$ is concave in $z$ for $z > 0$
(since $\theta < 1/2$ implies $z^\star > 0$).
Therefore, by Jensen's inequality:
\begin{equation}
  \mathcal{F}_{\bar{W}}(a)(x)
  \geq \int_\Pi
    \mathcal{F}_{W^\pi}(a)(x)\,d\nu(\pi)
  \quad \text{a.e.}
\end{equation}
Initialize $a^{(0)} \equiv 1/2$. Since
$J_1(W^\pi) = J_1(W^\star)$ for all $\pi$ and the
fixed-point iteration converges to a unique limit
under \cref{ass:contraction}, we have
$a^{(k)}_{W^\pi}(x) = a^{(k)}_{W^\star}(x)$ for all
$\pi$ and all $k \geq 0$. By induction, using
monotonicity of $\mathcal{F}$ in $a$ and concavity
of $\mathcal{F}$ in $W$:
\begin{multline}
  a^{(k+1)}_{\bar{W}}(x)
  = \mathcal{F}_{\bar{W}}(a^{(k)}_{\bar{W}})(x)
  \geq \mathcal{F}_{\bar{W}}(a^{(k)}_{W^\star})(x) \\
  \geq \int_\Pi
    \mathcal{F}_{W^\pi}(a^{(k)}_{W^\star})(x)
    \,d\nu(\pi)
  = a^{(k+1)}_{W^\star}(x).
\end{multline}
Taking $k \to \infty$ and using convergence under
\cref{ass:contraction}:
\begin{equation}
  a_\infty(\bar{W})(x)
  \geq a_\infty(W^\star)(x) \quad \text{a.e.}
\end{equation}
Integrating gives $J_1(\bar{W}) \geq J_1(W^\star)$.
Since $W^\star$ is a local maximizer and $\bar{W}$
achieves at least the same objective value, we may
work with $\bar{W}$ without loss of generality
throughout the remainder of the proof.

Since $\bar{W}$ is an SBM, its local externality
$\bar{z}(x) = \int_0^1 \bar{W}(x,y)
(\bar{a}(y)-\theta)\,dy$ is constant within each
type, so the steady state $\bar{a}(x) =
\sigma_{\beta(x)}(\bar{z}(x))$ is piecewise constant
on $[0,1]$. Since $\bar{z}(x)$ is
constant within each type, $\sigma_{\beta(x)}'
(\bar{z}(x))$ is also constant within each type, and
the adjoint equation~\eqref{eq:adjoint} then gives
$\mu^\star(x)$ constant within each type. This
establishes~\eqref{eq:constant_types}.

Finally, substituting~\eqref{eq:constant_types} into
 function~\eqref{eq:switching},
$G^\star(x,y)$ takes exactly three values on the
blocks $[0,\ell)\times[0,\ell)$,
$([0,\ell)\times[\ell,1])\cup([\ell,1]\times[0,\ell))$,
and $[\ell,1]\times[\ell,1]$, given
by~\eqref{eq:gmm}--\eqref{eq:ghh}. The 
condition~\eqref{eq:bangbang} then assigns a constant
value $w_{ij} \in [0,1]$ to each block, so $W^\star$
is an SBM as in~\eqref{eq:sbm}.
\end{proof}

\section{Greedy Edge Density Allocation}
\label{sec:allocation}

The implication of \cref{thm:sbm} is that every local maximizer of $J_1$
over $\mathcal{W}_\rho$ is an SBM graphon.
Therefore, the global maximizer is also an
SBM, and the search over $\mathcal{W}_\rho$ can be restricted to SBMs without loss of optimality. The optimization problem then becomes to allocate the connectivity budget $\rho$ across the three blocks: machine-machine
(MM), human-machine (HM), and human-human
(HH) with masses $(1-\ell)^2$, $2\ell(1-\ell)$,
and $\ell^2$ respectively. 

By 
condition~\eqref{eq:bangbang}, a block $ij$ is fully
active ($w_{ij}=1$) when $G_{ij} > 2\lambda$,
partially active ($w_{ij} \in (0,1)$) when
$G_{ij} = 2\lambda$, and inactive ($w_{ij}=0$) when
$G_{ij} < 2\lambda$, where $\lambda$ is determined
by the density constraint $\iint W^\star\,dx\,dy =
\rho$. Therefore, as $\rho$ increases from $0$ to
$1$, blocks are activated in decreasing order of
their value $G_{ij}$.

Given the switching values $G_{MM}$, $G_{HM}$,
$G_{HH}$, rank the three blocks in decreasing order:
\begin{equation}
  G_{(1)} \geq G_{(2)} \geq G_{(3)},
\end{equation}
with corresponding masses $m_{(1)}$, $m_{(2)}$,
$m_{(3)}$. The constant $\lambda$ is the unique
value determined by identifying the marginal block
index $k^\star \in \{1,2,3\}$ satisfying:
\begin{equation}
  \label{eq:lambda}
  \sum_{k=1}^{k^\star-1} m_{(k)}
  < \rho
  \leq \sum_{k=1}^{k^\star} m_{(k)},
\end{equation}
so that $\lambda = G_{(k^\star)}/2$. The optimal
block values are then:
\begin{equation}
  \label{eq:allocation}
  w_{(k)} =
  \begin{cases}
    1
      & k < k^\star, \\
    \frac{1}{m_{(k^\star)}}\left(\rho - \displaystyle\sum_{i=1}^{k^\star-1}
           m_{(i)}\right)
      & k = k^\star, \\
    0
      & k > k^\star.
  \end{cases}
\end{equation}

\begin{remark}
The allocation~\eqref{eq:allocation} is akin to a
\emph{water-filling} algorithm \cite{cover2006elements}: the edge connectivity budget
$\rho$ fills blocks in order of their value $G_{ij}$, with the Lagrange multiplier
$\lambda$ playing the role of the water level.
Unlike classical water-filling where the objective
is concave and solutions are globally optimal, our problem is nonconvex. Therefore, the allocation
satisfies the KKT conditions of \cref{lem:first_order_condition}, and is 
a locally optimal SBM solution.
\end{remark}

\section{Numerical Results}
\label{sec:numerical}

We illustrate the optimal allocation algorithm of
\cref{sec:allocation} on a stochastic block
model with machines ($\beta_M = 3.99$) and humans
($\beta_H = 2$), coordination game parameter $\theta = 0.3$,
and equal population mass $\ell = (1-\ell)=0.5$. The block
masses are $m_{MM} = 0.25$, $m_{HM} = 0.50$,
$m_{HH} = 0.25$. In our numerical experiments, the machine rationality parameters are set to $\beta_M = 3.99$ and $\beta_H=2$ to denote that machines react more sharply to network externalities than humans in the graphon coordination game. The fact that $\beta_M<4$ guarantees that the equilibrium $a_\infty$ is unique. 

\subsection{Allocation and adoption across density budgets}

Assuming the ordering
$G_{MM} \geq G_{HM} \geq G_{HH}$ and applying the greedy
algorithm for $\rho \in \{0.1, 0.2, \ldots, 0.9\}$, we obtain
\cref{tab:vary_rho}, which shows the optimal block weights,
steady-state adoption probability of taking the risky action (hunting stag), and aggregate adoption~$J_1$
for each density budget. The
hypothesized ordering is self-consistent at every value
of~$\rho$, i.e., when the $G_{ij}$'s are computed using \eqref{eq:gmm}, \eqref{eq:ghm} and \eqref{eq:ghh}, the ordering is preserved. Therefore, a local optimum has been found.

\begin{table}[h!]
\centering
\caption{Connectivity allocation under as a function of
  the density budget~$\rho$, with
  $\ell=0.5$, $\theta=0.3$, $\beta_M=3.99$,
  $\beta_H=2$. The marginal block is marked
  with~$^\star$.}
\label{tab:vary_rho}
\begin{tabular}{|c|ccc|cc|c|}
\hline
$\rho$
  & $w_{MM}$ & $w_{HM}$ & $w_{HH}$
  & $a_M$ & $a_H$
  & $J_1$ \\
\hline
\hline
0.1
  & $0.400^\star$ & 0 & 0
  & 0.550 & 0.500 & 0.525 \\
0.2
  & $0.800^\star$ & 0 & 0
  & 0.628 & 0.500 & 0.564 \\
0.3
  & 1 & $0.100^\star$ & 0
  & 0.697 & 0.510 & 0.604 \\
0.4
  & 1 & $0.300^\star$ & 0
  & 0.731 & 0.532 & 0.631 \\
0.5
  & 1 & $0.500^\star$ & 0
  & 0.766 & 0.558 & 0.662 \\
0.6
  & 1 & $0.700^\star$ & 0
  & 0.803 & 0.587 & 0.695 \\
0.7
  & 1 & $0.900^\star$ & 0
  & 0.838 & 0.619 & 0.729 \\
0.8
  & 1 & 1 & $0.200^\star$
  & 0.861 & 0.653 & 0.757 \\
0.9
  & 1 & 1 & $0.600^\star$
  & 0.873 & 0.692 & 0.782 \\
\hline
\end{tabular}
\end{table}

The water-filling algorithm proceeds as follows. The
machine-machine block receives all the budget until it
saturates at $\rho = m_{MM} = 0.25$, then the
human-machine block fills from $\rho = 0.25$ to
$\rho = m_{MM} + m_{HM} = 0.75$, and beyond
that value the human-human block receives the remaining
edges. These two thresholds correspond to phase
transitions in the optimal allocation, where the
marginal block shifts from~MM to~HM and from~HM
to~HH, respectively.
\cref{fig:waterfilling} displays the block weights
as continuous functions of~$\rho$, with shaded
regions indicating the marginal block and dashed
vertical lines at the phase transitions.

\begin{figure}[t]
  \centering
  \includegraphics[width=\columnwidth]{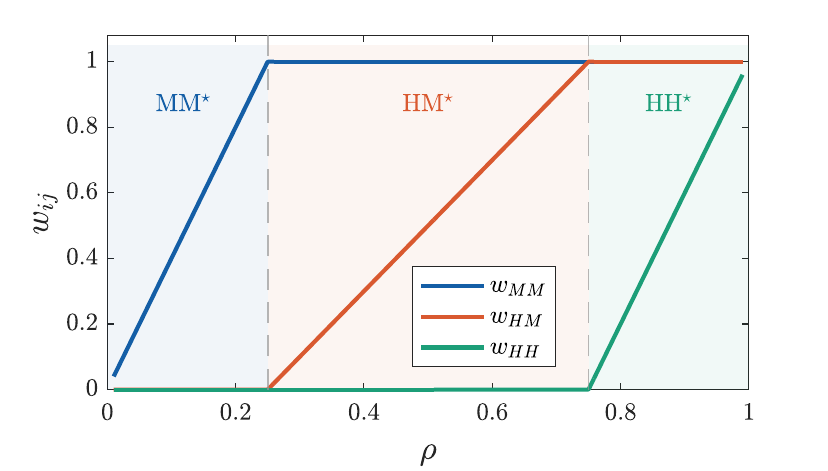}
  \caption{Optimal block weights $w_{MM}$, $w_{HM}$,
    $w_{HH}$ as a function of the density
    budget~$\rho$ for $\ell=0.5$. Phase
    transitions occur at $\rho = 0.25$ and
    $\rho = 0.75$.}
  \label{fig:waterfilling}
\end{figure}

The machine's adoption probability increases as the MM block
becomes more connected, reaching $a_M = 0.873$ by $\rho = 0.9$ from
a baseline of~$0.5$\footnote{The baseline of $0.5$ corresponds to completely irrational behavior, where the agent makes a decision by flipping a fair coin.}. Human adoption remains at the
baseline until the HM block begins to fill at
$\rho = 0.25$, after which it increases with $\rho$ until $a_H=0.692$. Notice that the HH block
does not receive links
until $\rho = 0.75$, which means that it is better to connect humans to the more rational machines before connecting humans to humans. At $\rho = 0.9$, aggregate
adoption reaches $J_1 = 0.78$, a gain of $28$
percentage points over a completely disconnected baseline, i.e., no externalities.
The quantities $a_M$, $a_H$, and $J_1$ are plotted as functions of~$\rho$ in \cref{fig:adoption}.

\begin{figure}[t]
  \centering
  \includegraphics[width=\columnwidth]{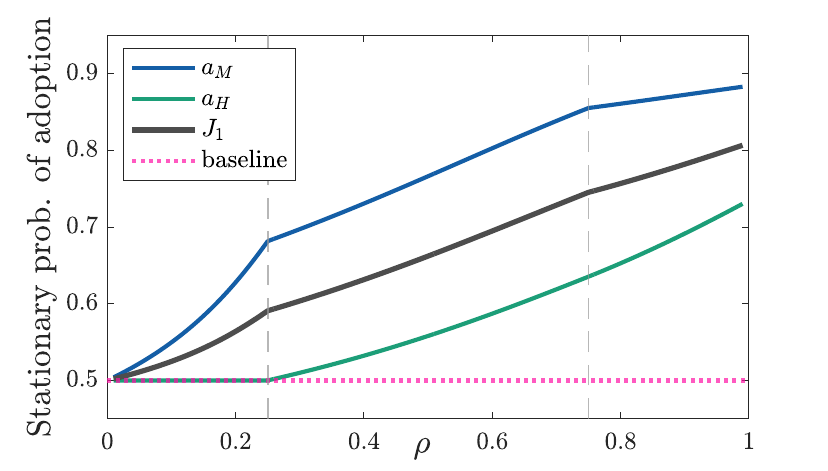}
  \caption{Steady-state activities $a_M$, $a_H$, and
    aggregate adoption $J_1$ under the optimal
    allocation.  The horizontal dotted line marks the
    baseline $a = 0.5$.}
  \label{fig:adoption}
\end{figure}

\subsection{Prioritizing human connectivity rather than machines}

To quantify the value of the optimality conditions,
\cref{fig:comparison} compares the optimal allocation
with two alternatives that use the same density
budget~$\rho$: a uniform (egalitarian) allocation
($w_{MM} = w_{HM} = w_{HH} = \rho$) and an
allocation that fills blocks in the order
$G_{HH} \geq G_{HM} \geq G_{MM}$, prioritizing
human-human connections.

\begin{figure}[t]
  \centering
  \includegraphics[width=\columnwidth]{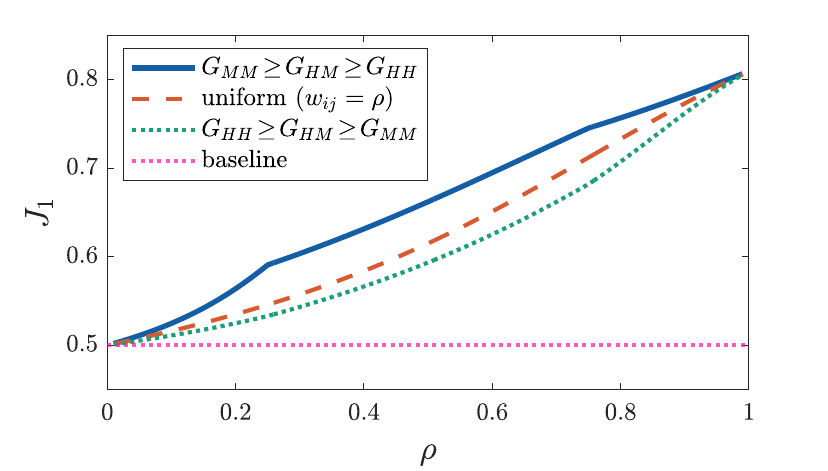}
  \caption{Aggregate adoption $J_1$ as a function
    of~$\rho$ under three allocation strategies:
    prioritizing machine connections ($G_{MM}\geq G_{MH} \geq G_{HH}$), uniform ($w_{MM}= w_{MH}= w_{HH}=\rho$), and prioritizing human connections $G_{HH}\geq G_{MH} \geq G_{MM}$.}
  \label{fig:comparison}
\end{figure}

The allocation that prioritizes machine connectivity strictly dominates both
alternatives for all $\rho \in (0,1)$. At
$\rho = 0.5$, the optimal allocation achieves
$J_1 = 0.662$, compared to $0.616$ for the uniform
allocation and $0.594$ for the reversed allocation.
The gap is largest at intermediate values of $\rho$. At very low
$\rho$ all blocks are mostly disconnected, and at $\rho$ near~$1$ all blocks are
completely connected regardless of priority ordering. Moreover, the allocation that prioritizes human connectivity is only self-consistent at one of the
nine grid points tested
($\rho = 0.6$), which confirms that the priority ordering
$G_{MM} \geq G_{HM} \geq G_{HH}$ is likely the sole
consistent regime for the chosen parameters.

\vspace{5pt}

\begin{remark}
  For the parameter values considered, the ordering
  $G_{MM} \geq G_{HM} \geq G_{HH}$ is the only
  self-consistent water-filling priority pattern across all
  tested values of $\rho$ and $\ell$. The
  alternative permutations of $(G_{MM}, G_{HM},
  G_{HH})$ fail the consistency verification step in all cases we have tried. This observation also reflects the large gap in rationality with $\beta_M \approx 2\beta_H$, which implies that machines respond to network externalities more strongly than humans, making MM connections
  a higher-return investment for the system designer interested in promoting coordination.
\end{remark}

\subsection{Finite human-machine network}

\cref{fig:graphon_network} illustrates the (locally) optimal SBM
graphon and a finite network  sampled from it for
$\rho = 0.35$ and $\ell = 0.5$. The graphon
(\cref{fig:graphon_network}, top) is a piecewise
constant $2 \times 2$ block matrix with
$w_{MM} = 1.00$, $w_{HM} = 0.20$, and
$w_{HH} = 0.00$. The sampled network
(\cref{fig:graphon_network}, bottom) consists of
$n = 20$ nodes ($10$ machines and $10$
humans) with edges drawn independently according to
the graphon probabilities. The resulting graph
exhibits the structure from the optimal graphon: a complete machine
sub-network ($45$ edges), a sparse cross-type
connections ($23$ out of $100$ possible human-machine
edges), and no human-human connections.

\begin{figure}[t]
  \centering
  \begin{minipage}[t]{0.38\textwidth}
    \centering
    \includegraphics[width=\columnwidth]{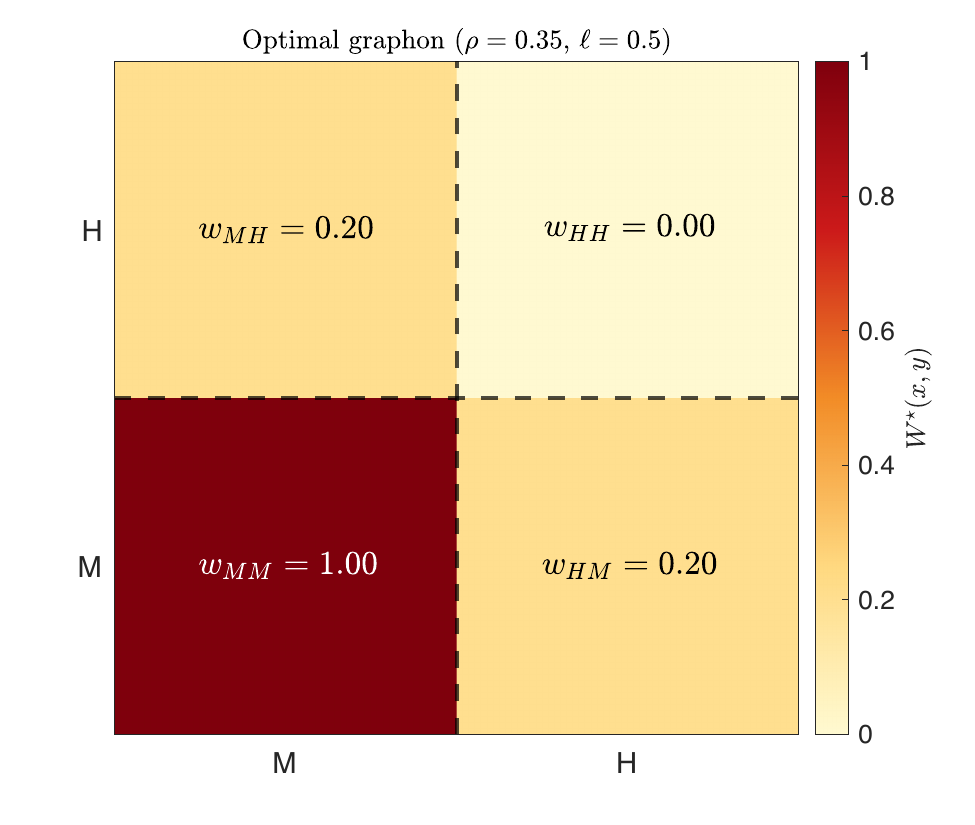}
  \end{minipage}
  \hfill
  \begin{minipage}[t]{0.5\textwidth}
    \centering
    \includegraphics[width=0.9\textwidth]{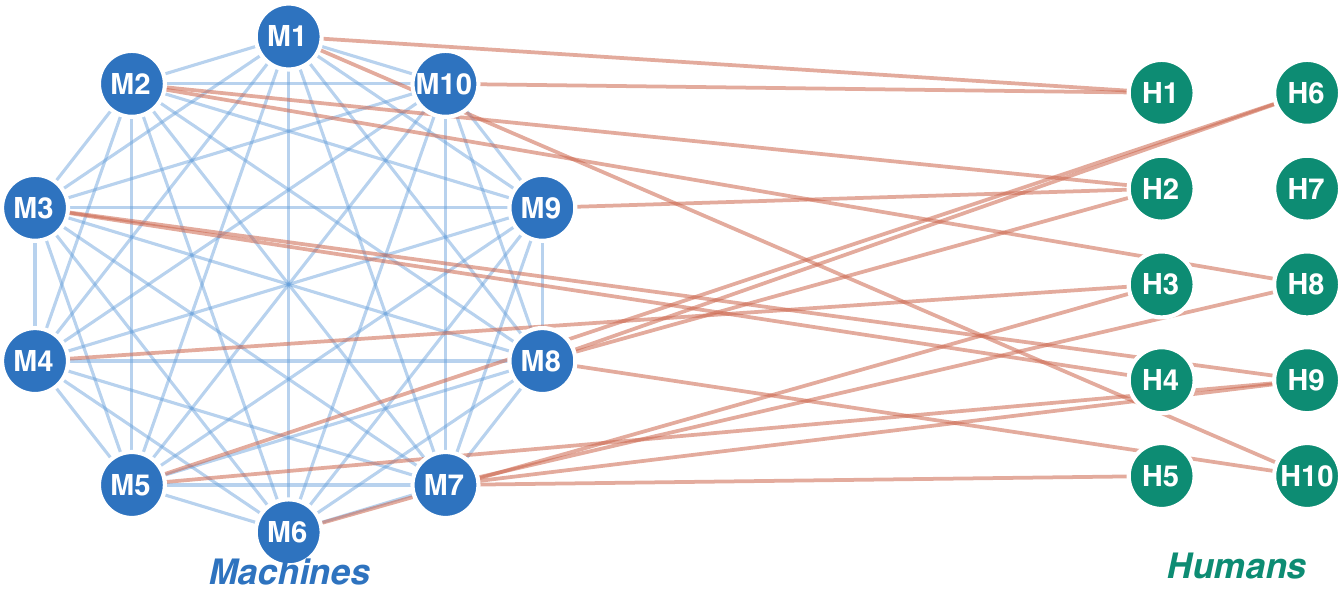}
  \end{minipage}
  \caption{Left: optimal graphon $W^\star$ for
    $\rho = 0.35$, $\ell = 0.5$. Right: a network
    of $n = 20$ nodes sampled from $W^\star$.
    Blue edges connect machines, coral edges link
    machines to humans, and the human-human block is
    empty.}
  \label{fig:graphon_network}
\end{figure}

\section{Conclusions}
\label{sec:conclusions}

We studied the problem of designing a network topology to
maximize coordination in a human-machine population with heterogeneous
bounded rationalities. Modeling the population as a continuum
of agents playing the stag hunt game under logit learning
dynamics, we formulated the problem as an infinite-dimensional
optimization over the space of graphons with a fixed edge
density.

Our main contributions are threefold. First, we derived
variational first-order optimality conditions for the graphon
design problem. Second,
we proved that the optimal
steady-state and adjoint functions are constant within each
type, and consequently that for a
bimodal rationality profile there exists an optimal graphon within the ensemble of stochastic block
models. Third, we obtained a greedy algorithm akin to water-filling to allocate the edge densities across blocks. 

Several research directions remain open. The global optimality of the SBM over the entire class $\mathcal{W}_\rho$ is conjectured but is not yet proved. Extension to continuous rationality profiles $\beta(x)$ requires characterizing the level sets of the $G(x,y)$ function in the two-dimensional  space $(\beta(x),a^\star(x))$, leading to more complex optimal graphons. Finite-sample performance guarantees and distributed algorithms converging to the optimal topology are also important next steps.

\bibliography{references}
\bibliographystyle{ieeetr}

\end{document}